\documentclass[11pt]{amsproc}
\usepackage{mathrsfs, color}
\usepackage{stmaryrd}
\usepackage{cases}
\usepackage{amssymb}
\usepackage{amsmath}
\usepackage{amsfonts}
\usepackage{graphicx}
\usepackage{amsmath,amstext,amsbsy,amssymb}
\newtheorem{theorem}{Theorem}[section]
\newtheorem{lemma}[theorem]{Lemma}

\theoremstyle{definition}

\theoremstyle{remark}

\numberwithin{equation}{section} \errorcontextlines=0

\newcommand{\la}{\lambda}

\newcommand{\ep}{\epsilon}

\newcommand{\br}{\mathbf r}
\newcommand{\bs}{\mathbf s}
\newcommand{\bz}{\mathbf z}
\begin{document}
\title[Quantum Discord of Certain Two-Qubit States ]
{Quantum Discord of Certain Two-Qubit States}%Quantum Discord of Non-X-States }
\author{Jianming Zhou}
\address{Zhou: Jianghan University, Wuhan, Hubei 430056, China}

\author{Xiaoli Hu*}
\address{Hu: Jianghan University, Wuhan, Hubei 430056, China}
\email{xiaolihumath@163.com}

\author{Naihuan Jing}
\address{Jing: Department of Mathematics,
   North Carolina State University,
   Raleigh, NC 27695, USA}
\email{jing@math.ncsu.edu}
%\footnotetext[1]{Corresponding author.}
\thanks{*Corresponding author: Xiaoli Hu}

\begin{abstract}
Quantum discord is an effective measure of quantum correlation introduced by Olliver and Zurek. We evaluate analytically the quantum discord for a large family of non-X-states. Exact solutions of the quantum discord are obtained of the four parametric space for non-X-states. Dynamic behavior of the quantum discord is also explored under the action of the Kraus operator.

%PACS numbers: 03.65.Ud, 03.67.Mn, 02.10.Yn, 02.10.Xm
\end{abstract}
\maketitle
\section{Introduction}

Quantum correlations are one of the fundamental features in quantum computation and quantum information. Among various measurements of
quantum correlations, quantum discord has been studied as a particularly important problem \cite{01, 02, 03}.

The notion of quantum discord was introduced by Olliver and Zurek \cite{1} to measure the difference of two natural quantum extensions of the classical mutual information \cite{101,102,103,104,105,106,107}. The quantum classical mutual information is usually used to quantify the total correlations with well documented  basic physical significance \cite{ 4,5,6,7}.  Let's consider the bipartite quantum state $\rho$, the quantum mutual information is defined as
\begin{equation}\label{e:quntum-inf}
\mathcal{I}(\rho):=S(\rho^a)+S(\rho^b)-S(\rho),
\end{equation}
where $S(\rho):=-\mathrm{Tr}\rho\log_2(\rho)$ is the von Neumann entropy of the quantum state.

In order to reveal the essence of quantum correlation, measurement entropy  based conditional density operators are used to study the classical correlation \cite{1}. The von Neumann measurement is  an entire set of projectors $\{B_k\}$ such that $\sum_kB_k=I$ and $B_jB_k=\delta_{jk}B_k$. If the measurement $\{B_k\}$ is executed locally on one side of the bipartite quantum state $\rho$, the quantum state is turned into
\begin{equation}\label{e:marginal}
\rho_k=\frac{1}{p_k}\mathrm{Tr}_b(I\otimes B_k)\rho(I\otimes B_k)
\end{equation}
with the probability $p_k=\mathrm{Tr}(I\otimes B_k)\rho(I\otimes B_k)$, and here $I$ is the identity operator for the party $a$. The quantum conditional entropy with respect to the measurement $\{B_k\}$ is defined as $$S(\rho|\{B_k\})=\sum_kp_kS(\rho_k),$$
then the quantum mutual information is defined as
$$\mathcal{I}(\rho|\{B_k\})=S(\rho_a)-S(\rho|\{B_k\})$$
and the classical correlation is measured in terms of the quantity
$$\mathcal{C}(\rho):=sup_{\{B_k\}}\mathcal{I}(\rho|\{B_k\}).$$
To compare the two quantum analogs of the classical quantum information %: $\mathcal{I}(\rho)$ and $\mathcal{C}(\rho)$,
the so-called quantum discord is defined as the difference of these two quantities:
\begin{equation}\label{e:quantdis}
\mathcal{Q}(\rho):=\mathcal{I}(\rho)-\mathcal{C}(\rho).
\end{equation}

There are considerable studies on quantum discords for various types of quantum states \cite{02}. %\cite{St}.
In \cite{L} an exact formula was
obtained for the Bell state. There are several well-known methods to compute the quantum discord for general X-states \cite{102}, and
in \cite{JY} exact and analytic formulas for the general X type states were given, and in this latter work some of the confusions in
previous computations of the quantum discord were clarified.
In \cite{9} it was shown that the quantum discord of 2-qubit is more robust than entanglement.
In \cite{10} the quantum discord dynamics of 2-qubit states in independent and common non-Markovian environments are evaluated. Ref. \cite{11}  provided an approach to compute one way quantum deficit of 2-qubit states. In \cite{12} an analytical formula of
quantum discord was presented for the two-qubit quantum state of rank-2 by studying its classical correlation (see also \cite{JY}). Despite all these progresses, it is still a difficult problem to
find exact formula of the quantum discord for the general bipartite state, for instance, the quantum discord of the non-X-type of two-qubit states with rank more than 2 is
unknown. In this paper, we study the quantum discord for the two-qubit states of all rank and derive exact formulas for several nontrivial cases.

We also study the dynamics of the quantum discord in this important case. We use the Kraus operators ${K_i}$ to discuss the behavior of the 2-qubit non-X-state $\rho$ through the phase damping channels, where $\sum\limits_{i}K_{i}^{\dag}K_i=1$. Under the phase damping $\rho$ changes into
\begin{equation}
\begin{split}
\tilde{\rho}=\sum_{i,j=1,2}K_{i}^A\otimes K_j^B\cdot\rho\cdot(K_i^A\otimes K_j^B)^\dag,
\end{split}
\end{equation}
where the Kraus operators can be defined as $K_1^{A(B)}=|0\rangle\langle 0|+\sqrt{1-\gamma}|1\rangle\langle1|$ and $K_2^{A(B)}=\sqrt{\gamma}|1\rangle\langle1|$ with the decoherence rate $\gamma\in[0,1]$.

This paper is organized into two parts. First we derive a formula of the quantum discord for a general quantum bipartite state in non-X-type
and then give the exact quantum discord for several nontrivial regions. In the last part we study how the quantum discord
behaves under the action of Kraus operators. Through this study one hopes to understand better the quantum discord for the general quantum state.

\section{Quantum discord for non-X-states}
The quantum discord for the Bell diagonal state was completely calculated by Luo \cite{L}, and the analytical expression of the general X-state quantum discord is obtained in \cite{JY}.
Here we consider a certain non-X states and compute its exact quantum discord. %there are no studies of the quantum discord for 2-qubit non-X-states.

Let $\rho$ be the following quantum state
\begin{equation}\label{2.1}
\begin{split}
   \rho=\frac{1}{4}(I\otimes I+{\bf r}\cdot\vec{\sigma}\otimes I+I\otimes {\bf s}\cdot\vec{\sigma}+ \sum_{i=1}^3c_i\sigma_i\otimes \sigma_i),
\end{split}
\end{equation}
where $\vec{\sigma}=(\sigma_1,\sigma_2,\sigma_3)$ is the vector of Pauli matrices, ${\bf r}=(r_1,r_2,r_3), {\bf s}=(s_1,s_2,s_3), {\bf c}=(c_1,c_2,c_3) \in \mathbb R^3$. Here $\sigma_i$ are normalized as $\mathrm{tr}(\sigma_i\sigma_j)=2\delta_{ij}$. It is clear that the coefficients $r_i, s_i$ can be confined within the internal $[-1, 1]$. The two marginal states of $\rho$ are given by
\begin{equation}\label{e:margin}
\rho^a=\mathrm{Tr}_b\rho=\frac{1}{2}(I+\mathbf r\cdot\vec{\sigma}), \qquad
\rho^b=\mathrm{Tr}_a\rho=\frac{1}{2}(I+\mathbf s\cdot\vec{\sigma}).
\end{equation}
Then the von Neumann entropy of the quantum marginal states are given by
\begin{equation}
\begin{split}
S(\rho^a)
=&1-\frac{1+|{\bf r}|}{2}log_2(1+|{\bf r}|)
    -\frac{1-|{\bf r}|}{2}log_2(1-|{\bf r}|),\\
S(\rho^b)=&1-\frac{1+|{\bf s}|}{2}log_2(1+|{\bf s}|)
   -\frac{1-|{\bf s}|}{2}log_2(1-|{\bf s}|).
\end{split}
\end{equation}

Let us introduce the entropic function
\begin{equation}H_{\ep}(x)=\frac12 (1+\ep+x)\log_2(1+\ep+x)+\frac12 (1+\ep-x)\log_2(1+\ep-x).
\end{equation}
It is easy to see that $H_{\ep}(x)$ is an even function. Also $\min H_{\ep}(x)= H_{\ep}(0)=(1+\ep)log_2(1+\ep)$ and $\max H_{\ep}(x)= H_{\ep}(\max|x|)$.
Thus the quantum mutual information of $\rho$ is obtained as
\begin{equation}
\mathcal{I}(\rho)
=2-H_{\ep=0}(|{\bf r}|)-H_{\ep=0}(|{\bf s}|)
+\sum_{i=1}^4\lambda_i\log_2(\lambda_i),
\end{equation}
where $\la_i (i=1,\cdots,4)$ are the eigenvalues of $\rho$.

Now let's turn to the second mutual information---the classical correlation $\mathcal{C}(\rho)$, which is defined with help of
the von Neumann measurements. As it is well known that $\{B_k=V|k\rangle\langle k|V^\dag, k=0,1\}$, where $V\in \mathrm{SU}(2)$,
parameterize the von Neumann measures.

Note that $\mathrm{SU}(2)$ is homeomorphic to the unit sphere, so any unitary matrix $V=tI+\sqrt{-1}\sum_{i=1}^3y_i\sigma_i$
where $t,y_i \in \mathbb R$ $(i=1,2,3)$ are on the unit sphere:
\begin{equation}
t^2+\sum_{i=1}^3y_i^2=1
\end{equation}

Under the unitary transformation the two marginal states of $\rho$ in \eqref{e:margin} are changed to
\begin{align}
\rho_0&=\frac{1}{2(1+\mathbf s\mathbf z)}[(1+\mathbf s\mathbf z)I+(\mathbf r+c\mathbf z)\cdot\vec{\sigma}],\\
\rho_1&=\frac{1}{2(1-\mathbf s\mathbf z)}[(1-\mathbf s\mathbf z)I+(\mathbf r-c\mathbf z)\cdot\vec{\sigma}]
\end{align}
with $p_0=\frac{1+\mathbf s\mathbf z}{2}, p_1=\frac{1-\mathbf s\mathbf z}{2}$ and the unit vector $\mathbf z=(z_1,z_2,z_3)$ is given by %the variables
$$z_1=2(-ty_2+y_1y_3),z_2=2(ty_1+y_2y_3),z_3=t^2+y_3^2-y_1^2-y_2^2.$$

The eigenvalues of $\rho_0$ and $\rho_1$ are seen to be
\begin{align}
\lambda_{\rho_0}^\pm&=\frac{1}{2(1+\mathbf s\mathbf z)}(1+\mathbf s\mathbf z\pm |\mathbf r+c\mathbf z|)\\
\lambda_{\rho_1}^\pm&=\frac{1}{2(1-\mathbf s\mathbf z)}(1-\mathbf s\mathbf z\pm |\mathbf r-c\mathbf z|)
\end{align}

Let
\begin{align}\label{G}
G(\bz)=-H_{\ep=0}(\bs\bz)+\frac{1}{2}H_{\ep=\bs\bz}(|\br+c\bz|)
+\frac{1}{2}H_{\ep=-\bs\bz}(|\br-c\bz|),
\end{align}
then the classical correlations can be given by
\begin{equation}
\begin{split}
\mathcal{C}(\rho)=&sup_{\{B_k\}}\mathcal{I}(\rho|\{B_k\})
=S(\rho^a)-sup\{\sum_{k=0,1}p_kS(\rho_k)\}\\
=&S(\rho^a)-sup\{\sum_{k=0,1}{p_k(\lambda_{\rho_k}^+\log_2 \lambda_{\rho_k}^++\lambda_{\rho_k}^-\log_2 \lambda_{\rho_k}^-)}\}\\
 =&-H_{\ep=0}(|\br|)+\max\{G(\bz)\}.
\end{split}
\end{equation}
The following result computes the quantum discord \eqref{e:quantdis} for some non-X states.
\begin{theorem}
When $\bs=0$ and $c_1=c_2=c_3=c$,
the quantum discord is
\begin{equation}\label{2.13}
\begin{split}
\mathcal{Q}(\rho)&=\frac{1}{2}H_{\ep=c}(|{\bf r}|)
+\frac{1}{2}H_{\ep=-c}(\sqrt{4c^2+|{\bf r}|^2})\\&-\frac{1}{2}[H_{\ep=0}(|\br|+|c|)
+H_{\ep=0}(||\br|-|c||)].
\end{split}
\end{equation}
In particular, when $c=|\br| \neq 0$, the quantum discord is
\begin{equation}
\begin{split}
\mathcal{Q}(\rho)=&=\frac{1}{4}(1-c+\sqrt{5}c)log_2(1-c+\sqrt{5}c)\\
&+\frac{1}{4}(1-c-\sqrt{5}c)log_2(1-c-\sqrt{5}c)\\
&-\frac{1}{4}(1-2c)log_2(1-2c);
\end{split}
\end{equation}
When $|\br| = |\bs| =0$,
 the quantum discord is
\begin{equation}
\mathcal{Q}(\rho)=\frac{1}{4}[(1-3c)log_2{(1-3c)}-2(1-c)log_2(1-c)+(1+c)log_2{(1+c)}].
\end{equation}
\end{theorem}
We first prove the following lemma. %before we prove the theorem.

\begin{lemma}\label{LE1}Let $\theta=|\br+c\bz|^2$, then  $\min\theta = (|{\bf r}|-|c|)^2$ and $\max\theta = (|{\bf r}|+|c|)^2$.
\end{lemma}
\begin{proof} Since $\theta = {\bf r}^2+c^2+2c(r_1z_1+r_2z_2+r_3z_3)$ and  $z_1^2+z_2^2+z_3^2=1$, so we consider %$y=2c(r_1z_1+r_2z_2+r_3z_3)$,
$$F(z_1,z_2,z_3,\mu)=2c(r_1z_1+r_2z_2+r_3z_3)+\mu(1-z_1^2-z_2^2-z_3^2),$$
where $\mu$ is a parameter. Then
$\frac{\partial F}{\partial z_i}=2cr_i-2\mu z_i=0, (i=1,2,3)$ implies that
$$\mu=\pm\sqrt{c^2(r_1^2+r_2^2+r_3^2)}=\pm |c||\br|.$$
When $\mu=|c||\br|$, then $z_i=\frac{r_i}{|\bf r|}$ and
the minimal value $\theta_{min} = (|{\bf r}|-|c|)^2$.
Similarly when $\mu=-|c||\br|$,
the maximal value $\theta_{max} = (|{\bf r}|+|c|)^2$.
\end{proof}

When $\bs=0, c_1=c_2=c_3=c$, the function $G(z)$ in \eqref{G} becomes
\begin{equation}G(\theta)=\frac{1}{2}H_{\ep=0}(\sqrt{\theta} )+\frac{1}{2}H_{\ep=0}(\sqrt{2(|{\bf r}|^2+c^2)-\theta} ).
\end{equation}
Meanwhile, the eigenvalues of $\rho$ in this case are
\begin{equation}
\lambda_{1,2}=\frac{1}{4}(1+c\pm|\br|);\quad
\lambda_{3,4}=\frac{1}{4}(1-c\pm\sqrt{4c^2+|\br|^2}).
\end{equation}
As $\rho$ is nonnegative, $(1+c)^2\geq \br^2$ and $(1-c)^2\geq 4c^2+|\br|^2$. Subsequently
$|\br|^2+c^2\leq 1$, therefore both $|\br|, |c|\leq 1$. Moreover, $|\br|- c\leq1$ and  $|\br|+ c\leq1$.
Now we can prove the theorem.
\begin{proof} It is obvious that
\begin{equation}\label{2.15}
G((|{\br}|+|c|)^2)=G((|{\br}|-|c|)^2)=\frac{1}{2}H_{\ep=0}(|{\br}|+|c|)+\frac{1}{2}H_{\ep=0}(||{\br}|-|c||)
\end{equation}
The derivative of $G(\theta)$ is equal to
\begin{equation}\label{2.16}
\begin{split}
\frac{\partial G(\theta)}{\partial \theta}&=\frac{1}{8}[\frac{1}{\sqrt{\theta}}\log_2\frac{1+\sqrt{\theta}}{1-\sqrt{\theta}}\\
&-\frac{1}{\sqrt{2(|{\bf r}|^2+c^2)-\theta}}log_2\frac{1+\sqrt{2(|{\bf r}|^2+c^2)
-\theta}}{1-\sqrt{2(|{\bf r}|^2+c^2)-\theta}}].
\end{split}
\end{equation}
Let $g(x)=\frac{1}{x}log_2\frac{1+x}{1-x}$, $x\in(0,1)$. The function $g(x)$ is strictly increasing as
$$\frac{\partial g(x)}{\partial x}=\frac{2}{xln2}(\sum_{n=0}^{\infty}\frac{-x^{2n}}{2n+1}+\sum_{n=0}^{\infty}x^{2n})>0.$$
So when $\theta>|{\bf r}|^2+c^2$, \eqref{2.16} implies that $\frac{\partial G(\theta)}{\partial \theta} > 0$, so $G(\theta)$ is an increasing function. Similarly, when $\theta<|{\bf r}|^2+c^2$, $G(\theta)$ is a decreasing function. Hence $G(\theta)$ has the minimal value at $\theta=|{\bf r}|^2+c^2$, and $$\max G(\theta)=G((|{\bf r}|+|c|)^2)=G((|{\bf r}|-|c|)^2).$$

In particular, if $|{\bf r}|=|c|\neq 0$,  then $\theta\in [0,2(|{\bf r}|^2+c^2)]$, we have
\begin{equation}
\begin{split}
\max G(\theta)&=G(0)=G(2(|{\bf r}|^2+c^2))=\frac{1}{2}H(2(|{\bf r}|^2+c^2))\\
&=\frac{1}{4}[(1+2c)log_2(1+2c)+(1-2c)log_2(1-2c)].
\end{split}
\end{equation}
If $|{\bf r}|=0$, $c_1=c_2=c_3=c$, the state $\rho$ in \eqref{2.1} degenerate to the Werner state. We have $\theta=c^2$, then
\begin{equation}
\max{G(\theta)}=G(c^2)=H(c)=\frac{1}{2}\big[(1+c)log_2(1+c)+(1-c)log_2(1-c)\big].
\end{equation}
\end{proof}
Similarly, when $|\br|=0$ and $c_1=c_2=c_3=c$, the quantum discord is
\begin{equation}\label{2.14}
\begin{split}
\mathcal{Q}(\rho)&=\frac{1}{2}H_{\ep=-c}(\sqrt{4c^2+|{\bf s}|^2})-\frac{1}{2}H_{\ep=-c}(|{\bf s}|).
\end{split}
\end{equation}
\begin{theorem} When $|\bs|=0, c_1=c_2=0$, the quantum discord
$\mathcal{Q}(\rho)=0$;
When $|\br|=0, c_1=c_2=0$, the quantum discord is
\begin{equation}
\mathcal{Q}(\rho)=H_{\ep=0}(\frac{|\bs|}{\sqrt{s_1^2+s_2^2+(c_3+ s_3)^2}}).
\end{equation}
\end{theorem}

\begin{theorem}
 When $\bs=0,c_3=0,c_1=c_2=c$, the quantum discord is
\begin{equation}
\begin{split}
\mathcal{Q}(\rho)&=\frac{1}{2}[H_{\ep=0}(\alpha_+)+H_{\ep=0}(\alpha_-)-
H_{\ep=0}(\beta_+)-H_{\ep=0}(\beta_-)],
\end{split}
\end{equation}
where $$\alpha_{\pm}=\sqrt{2c^2+r_1^2+r_2^2+r_3^2\pm2\sqrt{c^4+c^2r_1^2+c^2r_2^2}},$$
$$\beta_{\pm}=\sqrt{(r_1\pm\frac{r_1c}{\sqrt{r_1^2+r_2^2}})^2+(r_2\pm\frac{r_2c}{\sqrt{r_1^2+r_2^2}})^2+r_3^2},$$
\end{theorem}
Example 1. Consider $\rho$ with $s_1=0.1, s_2=0.2, s_3=0.2, c=0.3$. $\rho$ can be written in the form
\begin{equation}
\begin{split}
\rho=\left(                 %左括号
  \begin{array}{cccc}   %该矩阵一共3 列，每一列都居中放置
    0.375 & 0.025-0.05i & 0 & 0 \\  % 第一行元素
    0.025+0.05i & 0.125 & 0.15 & 0\\
    0 & 0.15 & 0.225 & 0.025-0.05i\\
    0 & 0 & 0.025+0.05i & 0.275 \\  % 第二行元素
  \end{array}
\right)
\end{split}
\end{equation}
The eigenvalues of $\rho$ are $\lambda_1=0.0073, \lambda_2=0.25, \lambda_3=0.3427, \lambda_4=0.4$ and the behavior of $G(\theta)$ is described in Fig.1. The quantum discord $\mathcal{Q}(\rho)=0.1058844$.
\begin{figure}[!t]
\centering
\includegraphics[width=2.5in]{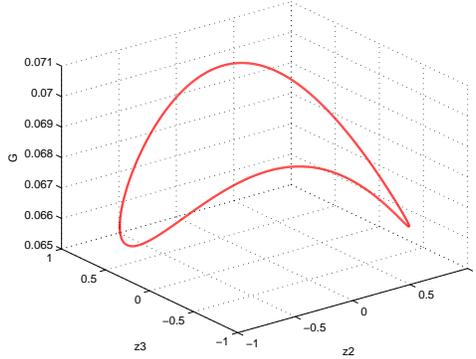}
\caption{The behavior of $G(\theta)$ for $\theta\in[0,0.0697]$ with parameters $s_1=0.1, s_2=0.2, s_3=0.2, c=0.3$}
\label{fig:rate_t50}
\end{figure}

Example 2. Let $s_1=s_2=s_3=0, r_1=0.1, r_2=0.2, r_3=0.25, c=0.3$, then $\rho$ is given by
\begin{equation}
\begin{split}
\rho=\left(
  \begin{array}{cccc}
    0.25&0&0.025-0.05i&0\\
    0&0.25&0.15&0.025-0.05i\\
    0.025+0.05i&0.15&0.25&0\\
    0&0.025+0.05i&0&0.25\\
  \end{array}
\right)
\end{split}
\end{equation}

The eigenvalues of $\rho$ are $\lambda_1=0.0815, \lambda_2=0.2315, \lambda_3=0.2685, \lambda_4=0.4185$, the quantum discord  $\mathcal{Q}(\rho)=0.0271$. The behavior of $G(\theta)$ is depicted in Fig.2. We can observe that the max of $G(\theta)$ is 0.2321.
\begin{figure}[!t]
\centering
\includegraphics[width=2.5in]{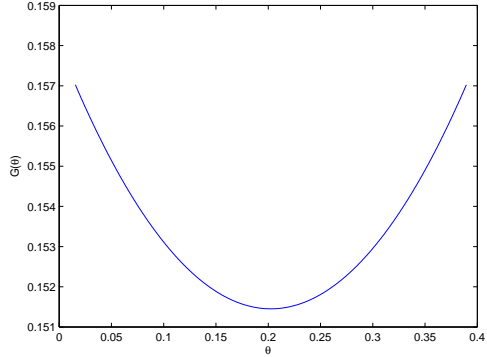}
\caption{The behavior of $G(\theta)$ with parameters $r_1=0.1, r_2=0.2, r_3=0.25, c=0.3$}
\label{fig:1}
\end{figure}

\section{Dynamics of quantum discord under phase damping channel}
In this section, we use the Kraus operators ${K_i}$ to discuss the behavior of the 2-qubit non-X-state $\rho$ through the phase damping channels [20], where $\sum\limits_{i}K_{i}^{\dag}K_i=1$. Under the phase damping $\rho$ is changed into
\begin{equation}
\begin{split}
\tilde{\rho}=\sum_{i,j=1,2}K_{i}^A\otimes K_j^B\cdot\rho\cdot(K_i^A\otimes K_j^B)^\dag
\end{split}
\end{equation}
where the Kraus operators can be defined as $K_1^{A(B)}=|0\rangle\langle 0|+\sqrt{1-\gamma}|1\rangle\langle1|$ and $K_2^{A(B)}=\sqrt{\gamma}|1\rangle\langle1|$ with the decoherence rate $\gamma\in[0,1]$. Therefore, under the phase damping $\rho$ in (2.1) becomes
\begin{equation}\label{1}
\begin{split}
\tilde{\rho}=&\frac{1}{4}[I\otimes I+\sum_{i=1,2}r_i\sqrt{1-\gamma}\sigma_i\otimes I+r_3\sigma_3\otimes I+I\otimes \sum_{i=1,2}s_i\sqrt{1-\gamma}\sigma_i
\\&+I\otimes s_3\sigma_3+c_3\sigma_3\otimes\sigma_3+\sum_{i=1,2}(1-\gamma)c_i\sigma_i\otimes\sigma_i].
\end{split}
\end{equation}

The two marginal states of $\tilde{\rho}$ are
\begin{align}
\tilde{\rho}^{a}&=\frac{1}{2}(I+\sum_{i=1,2}r_i\sqrt{1-\gamma}\sigma_i+r_3\sigma_3); \\
\tilde{\rho}^{b}&=\frac{1}{2}(I+\sum_{i=1,2}s_i\sqrt{1-\gamma}\sigma_i+s_3\sigma_3).
\end{align}
Thus the quantum mutual information of $\tilde{\rho}$ can be written as
\begin{equation}
\begin{split}
\mathcal{I}(\tilde{\rho})&=S(\tilde{\rho}^a)+S(\tilde{\rho}^b)-S(\tilde{\rho})\\
&=2-H_{\ep=0}(\sqrt{|{\br}|^2-\gamma r_1^2-\gamma r_2^2})-
H_{\ep=0}(\sqrt{|{\bs}|^2-\gamma s_1^2-\gamma s_2^2})\\
&+\sum_i^4\widetilde{\lambda}_ilog_2\widetilde{\lambda}_i,
\end{split}
\end{equation}
where $\widetilde{\lambda}_i (i=1,\cdots,4)$ are eigenvalues of $\widetilde{\rho}$.
Under the unitary transformation, the two marginal states of $\tilde{\rho}$ becomes
\begin{equation}
\begin{split}
\tilde{\rho}_{k}&=\frac{1}{2}[(1+(-1)^k\sqrt{1-\gamma}(s_1z_1+s_2z_2)+(-1)^k s_3z_3)I\\
&+\sum_i^2r_i\sqrt{1-\gamma}\sigma_i+r_3\sigma_3+(-1)^k c_3z_3\sigma_3+(-1)^k\sum_i^2c_i(1-\gamma)\sigma_iz_i];
\end{split}
\end{equation}
with $p_k=\frac{1}{2}(1+(-1)^k\sqrt{1-\gamma}(s_1z_1+s_2z_2)+(-1)^ks_3z_3)$ and $k=0,1$. The eigenvalues of $\tilde{\rho}_0$ and $\tilde{\rho}_1$ are given by
\begin{equation}
\lambda_{\tilde{\rho}_0}^\pm=\frac{1}{2(1+\varepsilon_{+})}(1+\varepsilon_{+}
\pm\sqrt{\zeta_{+}}); \ \
\lambda_{\tilde{\rho}_1}^\pm=\frac{1}{2(1+\varepsilon_{-})}(1+\varepsilon_{-}
\pm\sqrt{\zeta_{-}}),
\end{equation}
where $\varepsilon_{\pm}=\pm\sqrt{1-\gamma}(s_1z_1+s_2z_2)\pm s_3z_3$, $\zeta_{\pm}=(1-\gamma)[(r_1\pm \sqrt{1-\gamma}c_1z_1)^2+(r_2\pm\sqrt{1-\gamma}c_2z_2)^2]+(r_3\pm c_3z_3)^2$
The classical correlation $\mathcal{C}(\tilde{\rho})$ can be given by
\begin{equation}
\begin{split}
\mathcal{C}(\tilde{\rho})=&-H_{\ep=0}(\sqrt{|{\bf r}|^2-\gamma r_1^2-\gamma r_2^2})+\max \tilde{G}(\bz),
\end{split}
\end{equation}
where
\begin{equation}\label{3.9}
\tilde{G}(\bz)=-H_{\ep=0}(\varepsilon_{+})+
\frac{1}{2}(H_{\ep=\varepsilon_{+}}(\delta_+)+H_{\ep=\varepsilon_{-}}(\delta_-))
\end{equation}
with $\delta_{\pm}=\sqrt{(1-\gamma)[(r_1\pm\sqrt{1-\gamma}c_1z_1)^2+(r_2\pm\sqrt{1-\gamma}c_2z_2)^2]+(r_3\pm c_3z_3)^2}$.
Then the quantum discord of $\tilde{\rho}$ is
\begin{equation}\label{4.9}
\begin{split}
\mathcal{Q}(\tilde{\rho})&=2+\sum_i^4
\widetilde{\lambda}_ilog_2\widetilde{\lambda}_i-\max{\tilde{G}(\bz)}.
\end{split}
\end{equation}
 Under the phase damping channel, the Werner state $\rho$ becomes
\begin{equation}
\begin{split}
\tilde{\rho}=\frac{1}{4}[I\otimes I-c\sigma_3\otimes\sigma_3-c\sum_{i=1,2}(1-\gamma)\sigma_i\otimes\sigma_i].
\end{split}
\end{equation}
The eigenvalues of $\tilde{\rho}$ are $\tilde{\lambda}_1=\frac{1-c}{4}, \tilde{\lambda}_2=\frac{1-c}{4}, \tilde{\lambda}_3=\frac{1+3c-2c\gamma}{4}, \tilde{\lambda}_4=\frac{1-c+2c\gamma}{4}$.
The maximal value of $\tilde{G}(\bz)$ is written as
\begin{equation}
\begin{split}
\max {\tilde{G}(\bz)}=\frac{1+c}{2}log_2(1+c)+\frac{1-c}{2}log_2(1-c).
\end{split}
\end{equation}
Then the quantum discord of $\tilde{\rho}$ is given by
\begin{equation}
\begin{split}
\mathcal{Q}(\tilde{\rho})=&\frac{1+3c-2c\gamma}{4}log_2(1+3c-2c\gamma)\\
&+\frac{1-c+2c\gamma}{4}log_2(1-c+2c\gamma)\\&-\frac{(1+c)}{2}log_2(1+c).
\end{split}
\end{equation}
Thus,
\begin{equation}
\begin{split}
\mathcal{Q}(\rho)-\mathcal{Q}(\tilde{\rho})=&\frac{1-c}{4}log_2(1-c)+\frac{1+3c}{4}log_2(1+3c)\\
&-\frac{1+3c-2c\gamma}{4}log_2(1+3c-2c\gamma)\\
&-\frac{1-c+2c\gamma}{4}log_2(1-c+2c\gamma).
\end{split}
\end{equation}
Let $\mathcal{T}(c,\gamma)=\mathcal{Q}(\rho)-\mathcal{Q}(\tilde{\rho})$, the derivative of $\mathcal{T}(c,\gamma)$ is equal to
\begin{equation}
\begin{split}
\frac{\partial{\mathcal{T}}}{\partial\gamma}=\frac{c}{2}log_2\frac{1+3c-2c\gamma}{1-c+2c\gamma}.
\end{split}
\end{equation}
This is a strictly increasing function of $\gamma$. Thus, for fixed $c \in [0,1]$, the minimum of $\mathcal{T}(c,\gamma)$  is at $\gamma =0$. Therefore,when $\gamma\neq 0$, $\mathcal{Q}(\rho) > \mathcal{Q}(\tilde{\rho})$ . This shows that the quantum discord of  Werner state decreases under the phase damping channel.

When $c_1=c_2=0$ and $\bs=0$,the $\rho$ under the phase damping channel is given by
\begin{equation}
\begin{split}
\tilde{\rho}=\frac{1}{4}[I\otimes I+\sum_{i=1,2}r_i\sqrt{1-\gamma}\sigma_i\otimes I+r_3\sigma_3\otimes I+c_3\sigma_3\otimes\sigma_3].
\end{split}
\end{equation}
the eigenvalues of $\tilde{\rho}$ are $$\tilde{\lambda}_{1,2}=\frac{1\pm\sqrt{(r_1^2+r_2^2)(1-\gamma)+(c_3-r_3)^2}}{4},$$ $$ \tilde{\lambda}_{3,4}=\frac{1\pm\sqrt{(r_1^2+r_2^2)(1-\gamma)+(c_3+r_3)^2}}{4}.$$
The maximal value of $\tilde{G}(\bz)$ in \eqref{3.9} is given by
\begin{equation}
\begin{split}
\max {\tilde{G}(\bz)}= \frac{1}{2}(H_{\ep=0}(\varrho_+)+H_{\ep=0}(\varrho_-)),
\end{split}
\end{equation}
where $$\varrho_{\pm}=[1+\sqrt{(r_1^2+r_2^2)(1-\gamma)+(r_3\pm c_3)^2}].$$
Then $\mathcal{Q}(\tilde{\rho})=\mathcal{Q}(\rho)$.

When $|\bs|=0, c_3=0$ and $c_1=c_2=c$,the $\rho$ under the phase damping is described by
\begin{equation}
\begin{split}
\tilde{\rho}=\frac{1}{4}[I\otimes I+\sum_{i=1,2}r_i\sqrt{1-\gamma}\sigma_i\otimes I+r_3\sigma_3\otimes I+\sum_{i=1,2}(1-\gamma)c\sigma_i\otimes\sigma_i]
\end{split}
\end{equation}
We can also get that the eigenvalues of $\tilde{\rho}$ are $$\tilde{\lambda}_{1,2}=\frac{1}{4}[1\pm\sqrt{(1-\gamma)^2 2c^2+(1-\gamma)(r_1^2+r_2^2)+r_3^2+2\varsigma}],$$ $$\tilde{\lambda}_{3,4}=\frac{1}{4}[1\pm\sqrt{(1-\gamma)^2 2c^2+(1-\gamma)(r_1^2+r_2^2)+r_3^2-2\varsigma}],$$
where $\varsigma=\sqrt{c^4(1-\gamma)^4+(c^2r_1^2+c^2r_2^2)(1-\gamma)^3}$. Thus the maximal value of $\tilde{G}(\bz)$ in \eqref{3.9} is given by
\begin{equation}
\max {\tilde{G}(\bz)}= \frac{1}{2}(H_{\ep=0}(\xi_3)+H_{\ep=0}(\xi_4)).
\end{equation}
The  difference between quantum discord of the $\tilde{\rho}$ of  under the phase damping channel and quantum discord of $\rho$ is given by
\begin{equation}
\begin{split}
\mathcal{Q}(\rho)-\mathcal{Q}(\tilde{\rho})&=\frac{1}{2}\{H_{\ep=0}(\alpha_+)+H_{\ep=0}(\alpha_-)- (H_{\ep=0}(\beta_+)+H_{\ep=0}(\beta_-))
\\&-[H_{\ep=0}(\mu_+)+H_{\ep=0}(\mu_-)-(H_{\ep=0}(\sigma_+)+H_{\ep=0}(\sigma_-))]\}
\end{split}
\end{equation}
Where
$$\alpha_{\pm}=\sqrt{2c^2+r_1^2+r_2^2+r_3^2\pm2\sqrt{c^4+c^2r_1^2+c^2r_2^2}},$$
$$\beta_{\pm}=\sqrt{(r_1\pm\frac{r_1c}{\sqrt{r_1^2+r_2^2}})^2+(r_2\pm\frac{r_2c}{\sqrt{r_1^2+r_2^2}})^2+r_3^2},$$
$$\mu_{\pm}=\sqrt{2c^2(1-\gamma)+(r_1^2+r_2^2)(1-\gamma)+r_3^2\pm2\varsigma},$$
$$\sigma_{\pm}=\sqrt{(1-\gamma)((r_1\pm\sqrt{1-\gamma}\frac{r_1}{\sqrt{r_1^2+r_2^2}})^2+(r_2\pm\sqrt{1-\gamma}\frac{r_2}{\sqrt{r_1^2+r_2^2}})^2)+r_3^2}.$$

Now, we consider the state $\rho$ in Example 2. The state $\rho$ under phase damping channel is given by
\begin{equation}
\begin{split}
\tilde{\rho}=\frac{1}{4}[I\otimes I+0.1\sqrt{1-\gamma}\sigma_1\otimes I+0.2\sqrt{1-\gamma}\sigma_2\otimes I+0.3\sigma_3\otimes I\\+
0.25\sigma_3\otimes \sigma_3+0.25(1-\gamma)\sigma_1\otimes \sigma_1+0.25(1-\gamma)\sigma_2\otimes \sigma_2].
\end{split}
\end{equation}
For $\gamma = 0.2$, the eigenvalues of $\tilde{\rho}$ are $\tilde{\lambda}_1=0.399691, \tilde{\lambda}_2=0.328613, \tilde{\lambda}_3=0.217934, \tilde{\lambda}_4=0.0537617$ and the difference $\mathcal{Q}(\rho)-\mathcal{Q}(\tilde{\rho})=0.0583$. For $\gamma = 0.7$, the eigenvalues of $\tilde{\rho}$ are $\tilde{\lambda}_1=0.391011, \tilde{\lambda}_2=0.288475, \tilde{\lambda}_3=0.220322, \tilde{\lambda}_4=0.100192$ and the difference $\mathcal{Q}(\rho)-\mathcal{Q}(\tilde{\rho})=0.1426 $.
It is easily seen that $\mathcal{Q}(\rho)-\mathcal{Q}(\tilde{\rho})$ is different when $\gamma$ is different.

\bigskip

\section{Conclusions}

The quantum discord is an important quantum correlation with interesting applications.
It measures the difference between two natural quantum analogs of the classical mutual information.
Its computation is usually hard and exact formulas are difficult to derive.
For the general non-X-type quantum state,
we have given an analytical solution of the quantum discord in terms of the maximum of a one variable
function.

We have shown that the quantum discord essentially {\color{red} follows} the similar pattern as the other
types of quantum states. As an example,
we have shown that the quantum discord in the non-X-type case can also be computed exactly in several interesting regions. Using an example, our method is demonstrated to be able to
solve general non-X-type quantum states. We also studied the dynamics of the quantum discord under the Kraus operators,
and we have explained that there are cases the quantum discord is invariant under the process, while there are also
examples the quantum discord is changed. This is basically similar to the other situations.

In summary, the problem of the quantum discord for the general bipartite states {\color{red} follows} the similar pattern either
in the X-type or non-X-type.

\centerline{\bf Acknowledgments}
The corresponding author
gratefully acknowledges the partial support of NSFC grants 11426116, 11501251 and 11871325 during this work. The third author is supported
by NSFC grant 11531004 and Simons Foundation grant 523868.
 \vskip 0.1in

\bibliographystyle{amsalpha}

\end{document}